\documentclass{article}

\usepackage{arxiv}

\usepackage[utf8]{inputenc} % allow utf-8 input
\usepackage[T1]{fontenc}    % use 8-bit T1 fonts
\usepackage{fancyhdr}
\usepackage{lmodern}
\usepackage[colorlinks]{hyperref}  
\usepackage{url}            % simple URL typesetting
\usepackage{booktabs}       % professional-quality tables
\usepackage{amsfonts}       % blackboard math symbols
\usepackage{nicefrac}       % compact symbols for 1/2, etc.
\usepackage{microtype}      % microtypography
\usepackage{lipsum}

\usepackage{graphicx}
\usepackage{epsfig}
\usepackage{epstopdf}
\usepackage{subfigure}

\usepackage{enumerate}
\usepackage{amssymb}
\usepackage{amsmath}
\usepackage{amsthm}
\usepackage{enumitem}
\usepackage{mathrsfs}

\newcommand{\R}{\mathbb{R}}

\newcommand{\N}{\mathbb{N}}

\newcommand{\T}{\mathcal{T}} %Taylor
\newcommand{\Ra}{\mathcal{R}} %Rayleigh 
\newcommand{\CC}{\mathcal{C}} %salt Rayleigh
\renewcommand{\bar}{\overline}
\renewcommand{\rho}{\varrho}
\renewcommand{\phi}{\varphi}

\newtheorem{thm}{Theorem}[section]
\newtheorem{rem}{Remark}[section]

\title{The combined effects of rotation and anisotropy on double diffusive bi-disperse convection}

\author{  
F. Capone\thanks{Corresponding author.} \href{https://orcid.org/0000-0002-0672-999X}{\includegraphics[scale=0.1]{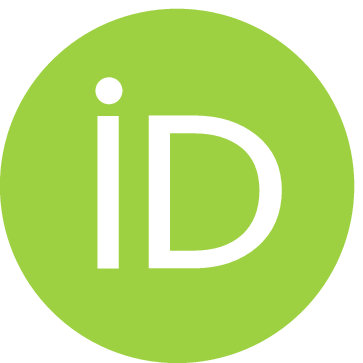}} \\ Dipartimento di Matematica e Applicazioni 'R.Caccioppoli' \\ Universit\'a degli Studi di Napoli Federico II \\ Via Cintia, Monte S.Angelo, 80126 Napoli \\ Italy \\ 
\texttt{fcapone@unina.it} \\
\And 
Roberta De Luca\href{https://orcid.org/0000-0002-2109-7564}{\includegraphics[scale=0.1]{orcid.eps}} \\ Dipartimento di Matematica e Applicazioni 'R.Caccioppoli' \\ Università degli Studi di Napoli Federico II \\ Via Cintia, Monte S.Angelo, 80126 Napoli \\ Italy \\   
\texttt{roberta.deluca@unina.it} \\     
\And  
G. Massa\href{https://orcid.org/0000-0002-8401-9176}{\includegraphics[scale=0.1]{orcid.eps}} \\ Dipartimento di Matematica e Applicazioni 'R.Caccioppoli' \\ Università degli Studi di Napoli Federico II \\ Via Cintia, Monte S.Angelo, 80126 Napoli \\ Italy \\   
\texttt{giuliana.massa@unina.it} \\ }

\renewcommand{\shorttitle}{The effects of rotation and anisotropy on double diffusive bi-disperse convection}

\begin{document}
\maketitle

\begin{abstract}
In the present paper double-diffusive convection, taking into account Coriolis effects, in a horizontal layer of Brinkman-anisotropic bi-disperse porous medium is analysed. Via linear instability analysis, we found that convection can set in through stationary or oscillatory motions and the critical Rayleigh numbers for the onset of stationary secondary flow (steady convection) and overstability (oscillatory convection) are determined.  
\end{abstract}
\keywords{Bi-disperse Porous Media \and Rotating Layer \and Double diffusion \and Instability analysis \and Anisotropy}
%\subjclass[2010]{76E05, 76E06, 76E07, 76Exx, 76S05, 76Sxx}
%parallel shear flow: 76E05
%convection: 76E06
%rotation: 76E07
%hydrodynamic stability: 76Exx
%flow in porous media; filtration; seepage: 76S05, 76Sxx

\section{Introduction}
The onset of convection is a widely studied problem due to its theoretical and applicative implications \cite{NB, BSLibro}. Recently, many researchers are turning their attention to bi-disperse convection, i.e. to the analysis of the onset of convection in dual porosity materials, called bi-disperse porous media. A bi-disperse porous medium (BDPM) is a material characterized by two types of pores called macropores - with porosity $\phi$ - and micropores - with porosity $\epsilon$. Therefore, $(1- \phi) \epsilon$ is the fraction of volume occupied by the micropores, $\phi + (1- \phi) \epsilon$ is the fraction of volume occupied by the fluid, $(1- \epsilon)(1-\phi)$ is the fraction of volume occupied by the solid skeleton. In particular, the macropores are referred to as f-phase (fractured phase), while the remainder of the structure is referred to as p-phase (porous phase). The first refined mathematical model describing the onset of bi-disperse convection was proposed by Nield and Kuznetsov in \cite{NK2005-1, NK2005, NK2006}, in those papers the authors extended the Brinkman model to the case of a bi-disperse porous medium and analysed the onset of convection in a horizontal layer of BDPM heated from below. Hence, taking into consideration the discussion and the analysis made in \cite{NK2005-1, NK2005, NK2006}, in the present paper a Brinkman BDPM is considered, i.e. the Brinkman law is employed to derive both macro- and micro-momentum equations. Moreover, Nield and Kuznetsov proved that the critical Rayleigh number for the onset of bi-disperse convection is higher if compared to the critical Rayleigh number for the single porosity case, so dual porosity materials are better suited for insulation problems and thermal management problems. Thereby, this kind of materials offers a lot of variegated possibilities to design man-made materials for heat transfer problems, this is the reason why a theoretical analysis of the onset of bi-disperse convection under various physical assumptions is essential. \\ 
The problem of the onset of convection in rotating clear fluids and porous materials finds a large number of practical applications: food process industry, chemical process industry, centrifugal filtration processes and other rotating machineries (see \cite{chandr, Vadasz, Vadasz2, massa1, massa3} and the references therein). To obtain even more useful results for the above applications, a salt dissolved in the fluid can be considered, so simultaneous heat and mass fluxes take place in the layer \cite{mulone2, BSsale}.

Envisaging a rotating machinery constituted by an engineered anisotropic bi-disperse porous material \cite{brian1, brian3}, in this paper the onset of bi-disperse double-diffusive convection will be analysed: we assume that the rotating horizontal layer heated from below of anisotropic BDPM is filled by an incompressible fluid binary mixture. The paper is organized as follows. In Section \ref{model} the mathematical model is presented and the equations governing the evolutionary behaviour of the perturbation to the thermal conduction solution are derived. In Section \ref{linear} linear instability analysis is performed to find the instability thresholds for the onset of steady and oscillatory double-diffusive convection. In Section \ref{num} the instability thresholds are numerically analysed in order to display the influence of the fundamental physical parameters on the onset of convection. Section \ref{concl} is a concluding section that summarises all the obtained results.

\section{Mathematical model}\label{model}
Let us consider a reference frame $Oxyz$ with fundamental unit vectors ${\bf i},{\bf j},{\bf k}$ (${\bf k}$ pointing vertically upward) and $L=\R^2 \times [0,d]$ a horizontal layer occupied by a bi-disperse porous medium saturated by an incompressible fluid binary mixture at rest state and uniformly heated from below. The layer $L$ rotates about the vertical axis $z$, with constant angular velocity ${\bf \Omega}=\Omega {\bf k}$. Let us assume that there is local thermal equilibrium between the f-phase and the p-phase, i.e. $T^f=T^p=T$ \cite{BSGentile}. Moreover, the fluid-saturated bi-disperse porous medium is \textit{horizontally isotropic}.
Let the axes $(x,y,z)$ be the \textit{principal axes} of the permeability, so the macropermeability tensor and the micropermeability tensor are

$$ \begin{aligned} {\bf K}^f &= \text{diag}(K^f_x,K^f_y,K^f_z)=K^f_z \ {\bf K}^{f*}, \\ {\bf K}^p &= \text{diag}(K^p_x,K^p_y,K^p_z)=K^p_z \ {\bf K}^{p*}, \end{aligned}$$
$${\bf K}^{f*}=\text{diag}(k,k,1),\quad
{\bf K}^{p*}=\text{diag}(h,h,1)$$
where
$$k=\frac{K^f_x}{K^f_z}=\frac{K^f_y}{K^f_z}, \quad h=\frac{K^p_x}{K^p_z}=\frac{K^p_y}{K^p_z}.$$
To derive the governing equations, a Boussinesq approximation is employed: the density is constant except in the buoyancy forces due to the gravity ${\bf g}=-g {\bf k}$, where it has a linear dependence on temperature and concentration fields, i.e.
$$\rho=\rho_F[1-\alpha(T-T_0)+\alpha_C(C-C_0)],$$
$\alpha$ and $\alpha_C$ being the thermal and the salt expansion coefficient, respectively, while $\rho_F$ is reference constant density. \\
Extending the Brinkman model in order to taking into account the Coriolis terms due to the uniform rotation of the layer about $z$ for the micropores and the macropores, the governing system is \cite{NK2005, BSsale, coriolis}:
\begin{equation} \label{sist1}\!\!\!\!
\begin{cases}  {\bf v}^f\! =\! \mu^{-1} {\bf K}^f \! \cdot \! \Bigl [ - \zeta ({\bf v}^f\! -\! {\bf v}^p) \!-\! \nabla p^f \! + \! \rho_F \alpha g T {\bf k} \! - \! \rho_F \alpha_C g C {\bf k} \! -\! \dfrac{2 \rho_F \Omega}{\phi} {\bf k} \times {\bf v}^f \!+\! \tilde{\mu}_f \Delta {\bf v}^f \Bigr], \\  {\bf v}^p \!=\! \mu^{-1} {\bf K}^p \! \cdot \! \Bigl[ - \zeta ({\bf v}^p    \!-\! {\bf v}^f) \!-\! \nabla p^p \! + \! \rho_F \alpha g T {\bf k} \! - \! \rho_F \alpha_C g C {\bf k} \!-\! \dfrac{2 \rho_F \Omega}{\epsilon} {\bf k} \times {\bf v}^p \!+\! \tilde{\mu}_p \Delta {\bf v}^p \Bigr],  \\  \nabla \cdot {\bf v}^f = 0, \\  \nabla \cdot {\bf v}^p = 0, \\  (\rho c)_m T_{,t} + (\rho c)_F ({\bf v}^f + {\bf v}^p) \cdot \nabla T = k_m \Delta T, \\ \epsilon_1 \dfrac{\partial C}{\partial t} +  ({\bf v}^f + {\bf v}^p) \cdot \nabla C = \epsilon_2 \Delta C  \end{cases}
\end{equation}
where 
$$p^s=P^s-\frac{\rho_F}{2} \vert {\bf \Omega} \times {\bf x} \vert^2, \quad s=f,p$$ are the reduced pressures, with ${\bf x}=(x,y,z)$, ${\bf v}^s$ is the seepage velocity for $s=\{f,p\}$, $T$ and $C$ are the temperature and concentration fields, $\zeta$ is an interaction coefficient between the f-phase and the p-phase, $\mu$ is the fluid viscosity, $c$ is the specific heat, $k_s$ is the thermal conductivity for $s=\{f,p\}$, $k^s_C$ = salt diffusivity for $s=\{f,p\}$, 
$$\begin{aligned}
&(\rho c)_m=(1-\phi)(1-\epsilon)(\rho c)_{sol}+\phi (\rho c)_f+\epsilon (1-\phi)(\rho c)_p, \\ &k_m=(1-\phi)(1-\epsilon)k_{sol}+\phi k_f+\epsilon (1-\phi)k_p, \\ &\epsilon_1=\phi+\epsilon(1-\phi), \ \epsilon_2=\phi k_C^f + \epsilon (1-\phi) k_C^p, \end{aligned}$$ 
(the subscripts $sol$ and $m$ are referred to the solid skeleton and to the medium). Since we are considering a single temperature BDPM and since macropores and micropores are saturated by the same mixture, we expect that $(\rho c)_f=(\rho c)_p=(\rho c)_F$, hence $(\rho c)_m=(1-\phi)(1-\epsilon)(\rho c)_{sol}+[\phi +\epsilon (1-\phi)](\rho c)_F$ \cite{fals-mul-BS}. \\
\noindent To (\ref{sist1})
the following boundary conditions are appended
\begin{equation} \label{BC1}
\begin{array}{l}
 {\bf v}^s \cdot {\bf n}=0\,, \quad s=\{f,p\}\qquad\mbox{on}\quad z=0,d,\\[2mm]
 T=T_L\,, \quad\mbox{on}\quad z=0\,,\qquad T=T_U\,, \quad\mbox{on}\quad z=d \\[2mm] C=C_L\,, \quad\mbox{on}\quad z=0\,,\qquad C=C_U\,, \quad\mbox{on}\quad z=d
\end{array}
\end{equation}
where ${\bf n}$ is the unit outward normal to the impermeable horizontal planes delimiting the layer and $T_L>T_U$, $C_L>C_U$, the layer being uniformly and simultaneously heated and salted from below. \\
\noindent
The problem $(\ref{sist1})$-$(\ref{BC1})$ admits
the stationary motionless solution (thermal conduction solution):
$${\bf \overline{v}}^f={\bf 0},\quad {\bf \overline{v}}^p={\bf 0}, \quad  \overline{T}=- \beta z + T_L, \quad  \overline{C}=- \beta_C z + C_L,$$
where $\beta=\dfrac{T_L-T_U}{d}$ is the temperature gradient, while $\beta_C=\dfrac{C_L-C_U}{d}$ is the concentration gradient.
Let us introduce a generic perturbation $\{ {\bf u}^f, {\bf u}^p, \theta, \gamma, \pi^f, \pi^p \}$ to the steady conduction solution, hence the evolutionary equations governing the perturbation fields are:
\begin{equation}\!\!\!\!
\begin{cases}  {\bf u}^f\! =\! \mu^{-1} {\bf K}^f \!\cdot\! \Bigl[ -\! \zeta ({\bf u}^f \!-\! {\bf u}^p) \!-\! \nabla \pi^f \!+\! \rho_F \alpha g \theta {\bf k} \!-\! \rho_F \alpha_C g \gamma {\bf k} \!-\! \dfrac{2 \rho_F \Omega}{\phi} {\bf k} \times {\bf u}^f \!+\! \tilde{\mu}_f \Delta {\bf u}^f \Bigr], \\  {\bf u}^p \!=\! \mu^{-1} {\bf K}^p \!\cdot\! \Bigl[ - \!\zeta ({\bf u}^p \!-\! {\bf u}^f)\! -\! \nabla \pi^p \!+\! \rho_F \alpha g \theta {\bf k} \!-\! \rho_F \alpha_C g \gamma {\bf k} \! -\! \dfrac{2 \rho_F \Omega}{\epsilon} {\bf k} \times {\bf u}^p \!+\! \tilde{\mu}_p \Delta {\bf u}^p \Bigr],  \\  \nabla \cdot {\bf u}^f = 0, \\  \nabla \cdot {\bf u}^p = 0, \\  (\rho c)_m \theta_{,t} + (\rho c)_F ({\bf u}^f + {\bf u}^p) \cdot \nabla \theta = (\rho c)_F \beta (w^f+w^p) + k_m \Delta \theta, \\ \epsilon_1 \dfrac{\partial \gamma}{\partial t} +  ({\bf u}^f + {\bf u}^p) \cdot \nabla \gamma = \beta_C (w^f+w^p) + \epsilon_2 \Delta \gamma \end{cases} \label{pertsist}
\end{equation}
%under the boundary conditions \cite{chandr}
%$$\displaystyle{\frac{\partial u^f}{\partial z}=\frac{\partial v^f}{\partial z}=\frac{\partial u^p}{\partial z}=\frac{\partial v^p}{\partial z}}=w^f=w^p=\theta=0 \ \text{on} \ z=0,d$$
where ${\bf u}^f=(u^f,v^f,w^f),\,{\bf u}^p=(u^p,v^p,w^p)$. To derive the dimensionless perturbed system, let us introduce the non-dimensional parameters
$$
%\begin{array}{l}
\! {\bf x}^{*} = \dfrac{{\bf x}}{d}, \ t^{*}=\dfrac{t}{\tilde{t}}, \ \theta^{*} = \dfrac{\theta}{\tilde{T}}, \ \gamma^{*} = \dfrac{\gamma}{{\tilde{C}}}, \ {\bf u}^{s*} = \dfrac{{\bf u}^s}{\tilde{u}}, \ \pi^{s*} = \dfrac{\pi^s}{\tilde{P}}, \quad \text{for} \ s=\{f,p\},$$
$$
\eta=\dfrac{\phi}{\epsilon}, \ \sigma=\dfrac{\tilde{\mu}_p}{\tilde{\mu}_f}, \ \gamma_1=\dfrac{\mu}{K^f_z \zeta}, \ \gamma_2=\dfrac{\mu}{K^p_z \zeta}, \ A=\dfrac{(\rho c)_m}{(\rho c)_F}, %\end{array}
$$
where the scales are given by
$$\tilde{u}=\frac{k_m}{(\rho c)_f d}, \ \tilde{t} = \frac{d^2 (\rho c)_m}{k_m}, \ \tilde{P} = \frac{\zeta k_m}{(\rho c)_f}, \ \tilde{T} = \sqrt{\frac{\beta k_m \zeta}{(\rho c)_f \rho_F \alpha g}}, \ \tilde{C} = \dfrac{k_m}{(\rho c)_F}\sqrt{\dfrac{\beta_C \zeta}{\epsilon_2 \rho_F \alpha_C g}}, $$
 and define the Lewis number $Le$, the Taylor number $\T$, the Darcy number $Da_f$, the thermal Rayleigh number $\Ra$, the chemical Rayleigh number $\CC$,
$$ Le=\dfrac{k_m}{\epsilon_2 (\rho c)_m}, \quad \T = \dfrac{2 \rho_F \Omega K^f_z}{\phi \mu}, \quad Da_f=\dfrac{\tilde{\mu}_f K^f_z}{d^2 \mu}, \quad \Ra = \sqrt{\dfrac{\beta d^2 (\rho c)_f \rho_F \alpha g}{k_m \zeta}}, \quad \CC = \sqrt{\dfrac{\beta_C d^2 \rho_f \alpha_C g}{\epsilon_2 \zeta}},$$
respectively. The resulting non-dimensional perturbation equations, dropping all the asterisks, are
\begin{equation}
\begin{cases}  \gamma_1 ({\bf K}^f)^{-1} {\bf u}^f\! +\! ({\bf u}^f \!-\! {\bf u}^p) \!=\! -\! \nabla \pi^f \!+\! \Ra \theta {\bf k} \!-\! \CC \gamma \textbf{k} \!-\! \gamma_1 \mathcal{T} {\bf k} \times {\bf u}^f \!+\! Da_f \gamma_1 \Delta {\bf u}^f, \\  \gamma_2 ({\bf K}^p)^{-1} {\bf u}^p \!-\! ({\bf u}^f \!-\! {\bf u}^p) \!=\! -\! \nabla \pi^p \!+\! \Ra \theta {\bf k} \!-\! \CC \gamma \textbf{k} \!-\! \eta \gamma_1 \mathcal{T} {\bf k} \times {\bf u}^p \!+\! Da_f \gamma_1 \sigma \Delta {\bf u}^p, \\  \nabla \cdot {\bf u}^f = 0, \\  \nabla \cdot {\bf u}^p = 0, \\  \theta_{,t} + ({\bf u}^f + {\bf u}^p) \cdot \nabla \theta = \Ra (w^f+w^p) +\Delta \theta, \\ \epsilon_1 Le \dfrac{\partial \gamma}{\partial t} + A \ Le ({\bf u}^f + {\bf u}^p) \cdot \nabla \gamma = \CC (w^f+w^p) + \Delta \gamma \end{cases} \label{pertubations}
\end{equation}
under the initial conditions $$\mathbf{u}^s(\mathbf{x},0)=\mathbf{u}^s_0(\mathbf{x})\,,\qquad \phi^s(\mathbf{x},0)=\phi_0(\mathbf{x})\,,\qquad \theta(\mathbf{x},0)=\theta_0(\mathbf{x})$$ with $\nabla\cdot\mathbf{u}_0^s=0,\, s=\{f,p\}$, and the stress-free boundary conditions \cite{chandr}
\begin{equation} \label{BOUNDCOND}
u^f_{,z}=v^f_{,z}=u^p_{,z}=v^p_{,z}=w^f=w^p=\theta=0 \quad \text{on} \ z=0,1.
\end{equation}
%The equations $(\ref{pertubations})$ are defined in $\{ (x,y,z,t) \in \R^4 | z \in (0,1), t>0 \}$. \\

\begin{rem}
According to experimental results, let us assume the perturbation fields being periodic functions in the horizontal directions $x,y$ of period $2 \pi / l$ and $2 \pi / m$, respectively, and let us denote by $$V=\Big[ 0,\frac{2 \pi}{l} \Big] \times \Big[ 0,\frac{2 \pi}{m} \Big] \times [0,1]$$ the periodicity cell. Moreover, let us assume that $\forall f \in  \{\nabla \pi^s,u^s,v^s,w^s,\theta,\gamma\}$ for $s=\{f,p\}$, $f \in W^{2,2}(V) \ \forall t \in \R^{+}$.
\end{rem}

\section{Onset of convection}\label{linear}
To determine the linear instability threshold for the onset of double diffusive convection, we linearise system (\ref{pertubations}) and seek for solutions ${\bf u}^f, {\bf u}^p, \theta, \gamma, \pi^f, \pi^p$ with time dependence like $e^{\bar\sigma t}$:
\begin{equation}
\begin{cases}  \gamma_1 ({\bf K}^f)^{-1} {\bf u}^f \!+\! ({\bf u}^f \!-\! {\bf u}^p) \!=\! -\! \nabla \pi^f \!+\! \Ra \theta {\bf k} \!-\! \CC \gamma \textbf{k} \!-\! \gamma_1 \mathcal{T} {\bf k} \times {\bf u}^f \!+\! Da_f \gamma_1 \Delta {\bf u}^f, \\  \gamma_2 ({\bf K}^p)^{-1} {\bf u}^p \!-\! ({\bf u}^f \!-\! {\bf u}^p) \!=\! -\! \nabla \pi^p \!+\! \Ra \theta {\bf k} \!-\! \CC \gamma \textbf{k} \!-\! \eta \gamma_1 \mathcal{T} {\bf k} \times {\bf u}^p \!+\! Da_f \gamma_1 \sigma \Delta {\bf u}^p, \\  \bar{\sigma} \theta = \Ra (w^f+w^p) +\Delta \theta \\ \epsilon_1 Le \bar{\sigma} = \CC (w^f+w^p) + \Delta \gamma  \end{cases} \label{pertubations2}
\end{equation}
Let us denote by
$$
\begin{array}{l}\Delta_1 f=f_{,xx}+f_{,yy}\,, \quad \Delta^m \equiv	\underbrace{ \Delta \Delta \cdots \Delta }_{m}\,, \qquad \omega^s_3=(\nabla \times {\bf u}^s) \cdot {\bf k},\,\, s=\{f,p\} \\ \qquad\qquad\qquad\qquad\bar a=\dfrac{\gamma_1}{k}+1\,,\quad \bar b=\dfrac{\gamma_2}{h}+1 \end{array}$$ 
and define the following operators 
\begin{equation}
A\equiv\bar{a}-Da_f \gamma_1 \Delta\,, \qquad B\equiv\bar{b}-Da_f \sigma \gamma_1 \Delta\,, \qquad \Psi\equiv(AB-1). \label{operatori}
\end{equation}
We compute the third components of curl and of double curl of
 $(\ref{pertubations})_{1,2}$, respectively given by
\begin{equation}\label{curl}
\begin{cases}
 A\omega_3^f -\omega_3^p=\gamma_1\T w_{,z}^f,\\
 -\omega_3^f+B\omega_3^p=\eta\gamma_1\T w_{,z}^p
 \end{cases}
\end{equation} 
 and
\begin{equation}
\mkern-18mu \begin{cases}
\! - \dfrac{\gamma_1}{k} w^f_{,zz} \! - \! \gamma_1 \Delta_1 w^f \! - \! \Delta w^f \! + \! \Delta w^p \! \! = \! \! - \! \Ra \Delta_1 \theta \!+\! \CC \Delta_1 \gamma, \! + \! \gamma_1 \T \omega^f_{3,z} \! - \! Da_f \gamma_1 \Delta^2 w^f \! , \\[2mm]
\! - \dfrac{\gamma_2}{h} w^p_{,zz} \! - \! \gamma_2 \Delta_1 w^p \! + \! \Delta w^f \! - \! \Delta w^p \! \! = \! \! - \! \Ra \Delta_1 \theta \!+\! \CC \Delta_1 \gamma, \! + \! \eta \gamma_1 \T \omega^p_{3,z} \! - \! Da_f \gamma_1 \sigma \Delta^2 w^p \! .
\end{cases}  \label{doublecurl}
\end{equation}
Applying the operator $B$ to (\ref{curl})$_1$, by virtue of (\ref{curl})$_2$, we obtain
$$
\Psi\omega_3^f=\gamma_1\T Bw_{,z}^f+\eta\gamma_1\T w_{,z}^p.
$$
This equation, together with that one obtained by applying the operator $\Psi$ to (\ref{curl})$_2$, leads to
\begin{equation} \label{rotore}
\begin{cases}
 \Psi\omega_3^f=\gamma_1\T B w_{,z}^f+\eta \gamma_1\T w_{,z}^p,\\
 \\
 \Psi B \omega_3^p=\gamma_1 \T B w_{,z}^f+\eta\gamma_1 \T AB w_{,z}^p.
\end{cases}
\end{equation}
Applying the operator $\Psi$
 to (\ref{doublecurl})$_1$ and $\Psi B$ to (\ref{doublecurl})$_2$, we obtain
\begin{equation}\label{abovesyst}
\begin{cases}  -\bar{a} \Psi w^f_{,zz} - \hat{\gamma}_1 \Psi \Delta_1 w^f + \Psi \Delta_1 w^p + \Psi w^p_{,zz} = \\  \qquad - \Ra \Psi \Delta_1 \theta +\CC \Psi  \Delta_1 \gamma + \gamma_1 \T \Psi \omega^f_{3,z} -Da_f \gamma_1 \Psi \Delta^2 w^f, \\[2mm] -\bar{b} \Psi B w^p_{,zz} -\hat{\gamma_2} \Psi B \Delta_1 w^p + \Psi B \Delta_1 w^f +\Psi B w^f_{,zz} = \\ \qquad - \Ra \Psi B \Delta_1 \theta + \CC \Psi B \Delta_1 \gamma + \eta \gamma_1 \T \Psi B \omega^p_{3,z} -Da_f \sigma \gamma_1 \Psi B \Delta^2 w^p, \
\end{cases}
\end{equation}
with $\hat{\gamma}_r=\gamma_r+1$, for $r=1,2$.
 \\
In view of $(\ref{rotore})$,  $(\ref{abovesyst})$ can be written as
\begin{equation}
\begin{cases}
[ -\bar{a}\Psi -(\gamma_1 \T)^2 B ] w^f_{,zz} - \hat{\gamma}_1 \Psi \Delta_1 w^f + \Psi \Delta_1 w^p + \\ \quad [ \Psi-\eta (\gamma_1 \T)^2] w^p_{,zz} +Da_f \gamma_1 \Psi \Delta^2 w^f= - \Ra \Psi \Delta_1 \theta+\CC \Psi \Delta_1 \gamma, \\[2mm] [ -\bar{b} \Psi B -(\eta \gamma_1 \T)^2 AB ] w^p_{,zz} - \hat{\gamma}_2 \Psi B \Delta_1 w^p + \Psi B \Delta_1 w^f + \\ \quad [ \Psi B-\eta (\gamma_1 \T)^2 B] w^f_{,zz} +Da_f \sigma \gamma_1 \Psi B \Delta^2 w^p= - \Ra \Psi B \Delta_1 \theta +\CC \Psi B \Delta_1 \gamma.
\end{cases} \label{finalPert}
\end{equation}
Consequently, we consider $(\ref{pertubations2})_{3,4}$, $(\ref{finalPert})_1$ and $(\ref{finalPert})_2$, i.e.:
\begin{equation}
\begin{cases} [ -\bar{a}\Psi -(\gamma_1 \T)^2 B ] w^f_{,zz} - \hat{\gamma}_1 \Psi \Delta_1 w^f + \Psi \Delta_1 w^p + \\ \quad [ \Psi - \eta (\gamma_1 \T)^2 ] w^p_{,zz} +Da_f \gamma_1 \Psi \Delta^2 w^f= - \Ra \Psi \Delta_1 \theta +\CC \Psi \Delta_1 \gamma , \\[3mm] [ -\bar{b} \Psi B -(\eta \gamma_1 \T)^2 AB ] w^p_{,zz} - \hat{\gamma}_2 \Psi B \Delta_1 w^p + \Psi B \Delta_1 w^f + \\ \quad [ \Psi B-\eta (\gamma_1 \T)^2 B] w^f_{,zz} +Da_f \sigma \gamma_1 \Psi B \Delta^2 w^p= - \Ra \Psi B \Delta_1 \theta +\CC \Psi B \Delta_1 \gamma, \\[3mm] \bar{\sigma} \theta= \Ra (w^f+w^p) +\Delta \theta, \\[3mm] \epsilon_1 Le \bar{\sigma} = \CC (w^f+w^p) + \Delta \gamma.
\end{cases} \label{sistema}
\end{equation}
Let us employ normal modes solutions in (\ref{sistema}) \cite{chandr}:
\begin{equation}
\begin{aligned} w^f &= W^f_0 \sin(n \pi z) e^{i(lx+my)}, \\ w^p &= W^p_0 \sin(n \pi z) e^{i(lx+my)}, \\ \theta &= \Theta_0 \sin(n \pi z) e^{i(lx+my)}, \\ \gamma &= \Gamma_0 \sin(n \pi z) e^{i(lx+my)},
\end{aligned}\label{17}
\end{equation}
$W^f_0,W^p_0,\Theta_0,\Gamma_0$ being real constants, so from (\ref{sistema}) it turns out that
\begin{equation}\label{24} \mkern-18mu \begin{cases} \Big[\! \Lambda_n e (A_1 M \!+\! \sigma f n^2 \pi^2) \!+\! \Lambda^2_n e (M e \sigma \!+\! B_1) \!+\! f \bar{b} n^2 \pi^2 \!+\! B_1 M  \!+\! \\ \ \ e^2 \Lambda_n^3 A_1 \!+\! e^3 \sigma \Lambda_n^4 \Big] \! W^f_0 \!+\! \Big[ \!-\! B_1 \Lambda_n \!-\! e A_1 \Lambda_n^2 \!-\! e^2 \sigma \Lambda_n^3 \!+\! \eta f n^2 \pi^2 \Big] \! W^p_0  \\ \ \ -\Ra a^2 \Big[ B_1 + e \Lambda_n A_1 +e^2 \sigma \Lambda_n^2 \Big] \Theta_0 
+ \CC a^2 \Big[ B_1 + e \Lambda_n A_1 +e^2 \sigma \Lambda_n^2 \Big] \Gamma_0 =0, \\ \\ \Big[ \Lambda_n (e \sigma n^2 \pi^2 \eta f - \bar{b} B_1)+\eta f n^2 \pi^2 \bar{b} -\Lambda_n^2 e C  \\ \ \ -\Lambda_n^3 e^2 \sigma (A_1+\bar{b}) -\Lambda_n^4 \sigma^2 e^3 \Big] W^f_0 + \\ \ \ \Big\{ \Lambda_n e (CN+\eta^2 f A_1 n^2 \pi^2) + \Lambda_n^2 e \sigma [e(A_1+\bar{b})N+\bar{b}B_1 + e \eta^2 f n^2 \pi^2]+\\ \ \ B_1 \bar{b} N +e^2 \Lambda_n^3 \sigma (C+e \sigma N)+\Lambda_n^4 e^3 \sigma^2 (A_1+\bar{b}) +\Lambda_n^5 e^4 \sigma^3+ \\ \ \ \eta^2 \! f n^2 \pi^2 \bar{a} \bar{b}  \Big\} W^p_0 \!-\! R a^2 \Big[ \bar{b} B_1 \!+\! e \Lambda_n C \!+\! \Lambda_n^2 e^2 \sigma (A_1 \!+\! \bar{b}) \!+\! e^3 \sigma^2 \Lambda_n^3 \Big]\! \Theta_0 \\ \ \ 
+ \CC a^2 \Big[ \bar{b} B_1 \!+\! e \Lambda_n C \!+\! \Lambda_n^2 e^2 \sigma (A_1 \!+\! \bar{b}) \!+\! e^3 \sigma^2 \Lambda_n^3 \Big]\! \Gamma_0 \!=\! 0, \\ \\ \Ra W^f_0+ \Ra W^p_0 - (\Lambda_n+\bar{\sigma}) \Theta_0=0, \\ \\ \CC W^f_0+ \CC W^p_0 - (\Lambda_n+ \epsilon_1 Le \bar{\sigma}) \Gamma_0=0, \end{cases}
\end{equation}
where $a^2=l^2+m^2$ and $\Lambda_n=a^2+n^2 \pi^2$, while
\begin{equation} \label{positions}
\begin{array}{l} A_1 = \sigma \bar a +\bar b, \,\,\, B_1 = \dfrac{\gamma_1}{k} \dfrac{\gamma_2}{h} + \dfrac{\gamma_1}{k} +\dfrac{\gamma_2}{h}, \,\,\, C =  \sigma (2 B_1 + 1) + {\bar b}^2, \\ \\ M = \dfrac{\gamma_1}{k} n^2 \pi^2 + \gamma_1 a^2 + \Lambda_n, \,\,\, N = \dfrac{\gamma_2}{h} n^2 \pi^2 + \gamma_2 a^2 + \Lambda_n, \\ \\ e = Da_f \gamma_1, \,\,\,\, f = (\gamma_1 \T)^2. \end{array}
\end{equation}
Setting
$$\begin{aligned} h_{11} =& \Lambda_n e (A_1 M+\sigma f n^2 \pi^2) + \Lambda^2_n e (M e \sigma +B_1) + f \bar{b} n^2 \pi^2 + B_1 M + \\ & e^2 \Lambda_n^3 A_1 +e^3 \sigma \Lambda_n^4, \\ h_{12} =& -B_1 \Lambda_n -e A_1 \Lambda_n^2 -e^2 \sigma \Lambda_n^3 +\eta f n^2 \pi^2, \\ h_{13} =& B_1 + e \Lambda_n A_1 +e^2 \sigma \Lambda_n^2, \\ h_{21} =& \Lambda_n (e \sigma n^2 \pi^2 \eta f - \bar{b} B_1)+\eta f n^2 \pi^2 \bar{b} -\Lambda_n^2 e C -\Lambda_n^3 e^2 \sigma (A_1+\bar{b}) -\Lambda_n^4 \sigma^2 e^3, \\ h_{22} =& \Lambda_n e (CN+\eta^2 f A_1 n^2 \pi^2) + \Lambda_n^2 e \sigma [e(A_1+\bar{b})N+\bar{b}B_1 + e \eta^2 f n^2 \pi^2]+ \\ & B_1 \bar{b} N+ e^2 \Lambda_n^3 \sigma (C+e \sigma N)+\Lambda_n^4 e^3 \sigma^2 (A_1+\bar{b}) +\Lambda_n^5 e^4 \sigma^3 +\eta^2 f n^2 \pi^2 \bar{a} \bar{b}, \\ h_{23} =& \bar{b} B_1 + e \Lambda_n C + \Lambda_n^2 e^2 \sigma (A_1+\bar{b}) + e^3 \sigma^2 \Lambda_n^3,
\end{aligned}$$
(\ref{24}) can be written as
\begin{equation}
\begin{cases} h_{11} W^f_0  +h_{12} W^p_0 - \Ra a^2 h_{13} \Theta_0 + \CC a^2 h_{13} \Gamma_0 = 0, \\ h_{21} W^f_0  +h_{22} W^p_0 - \Ra a^2 h_{23} \Theta_0 + \CC a^2 h_{23} \Gamma_0 = 0, \\ \Ra W^f_0 + \Ra W^p_0 - (\Lambda_n+\bar{\sigma}) \Theta_0 = 0, \\ \CC W^f_0+ \CC W^p_0 - (\Lambda_n+ \epsilon_1 Le \bar{\sigma}) \Gamma_0=0. \end{cases} \label{finale}
\end{equation}
Requiring zero determinant for system \eqref{finale}, we get:
\begin{equation}\label{PES1}
\Ra^2=\dfrac{\Lambda_n+\bar{\sigma}}{a^2} \dfrac{h_{11} h_{22} - h_{12} h_{21}}{h_{13} h_{22} - h_{12} h_{23} + h_{11} h_{23} - h_{21} h_{13}} + \CC^2 \dfrac{\Lambda_n+\bar{\sigma}}{\Lambda_n+\epsilon_1 Le \bar{\sigma}} \,.
\end{equation}
The growth rate is $\bar{\sigma}=\sigma_R + i \sigma_I$, so \eqref{PES1} is 
\begin{equation} \label{PES2}
\Ra^2 = Re(\Ra^2) + i \ Im(\Ra^2),
\end{equation}
where the real part and the imaginary part are respectively given by
\begin{equation}
\begin{aligned}
\!\! Re(\Ra^2) &\!=\!  \dfrac{ (\Lambda_n \!+\! \sigma_R) (h_{11} h_{22} \!-\! h_{12} h_{21})}{a^2(h_{12} h_{23} \!-\! h_{13} h_{22}\!-\! h_{11} h_{23} \!+\! h_{21} h_{13})} \!+\! \CC^2 \dfrac{(\Lambda_n + \sigma_R) (\Lambda_n + \epsilon_1 Le \sigma_R) + \epsilon_1 Le \sigma_I^2}{(\Lambda_n + \epsilon_1 Le \sigma_R)^2 + (\epsilon_1 Le \sigma_I)^2}, \\[2mm] \!\! Im(\Ra^2) &\!=\! \sigma_I \Bigl[ \dfrac{ h_{11} h_{22} \!-\! h_{12} h_{21}}{a^2(h_{12} h_{23} \!-\! h_{13} h_{22} \!-\! h_{11} h_{23} \!+\! h_{21} h_{13})} \!+\! \CC^2 \dfrac{\Lambda_n (1-\epsilon_1 Le)}{(\Lambda_n + \epsilon_1 Le \sigma_R)^2 + (\epsilon_1 Le \sigma_I)^2} \Bigr].
\end{aligned}
\end{equation}

\begin{thm}
If $\epsilon_1 Le \leq 1$, {\it the strong form of the principle of exchange of stability holds}, i.e. oscillatory convection cannot arise. 
\end{thm}
\begin{proof}
Let us underline that $h_{11} h_{22} - h_{12} h_{21}$ and $h_{12} h_{23} - h_{13} h_{22} - h_{11} h_{23} + h_{21} h_{13}$ are strictly positive. Since $\Ra^2$ is a real number, the imaginary part of \eqref{PES2} has to vanish: 
\begin{equation} \label{PES3}
\sigma_I \Bigl\{ ( h_{11} h_{22} \!-\! h_{12} h_{21}) [(\Lambda_n \!+\! \epsilon_1 Le \sigma_R)^2 \!+\! (\epsilon_1 Le \sigma_I)^2] \!+\!  \CC^2 a^2 (h_{12} h_{23} \!-\! h_{13} h_{22} \!-\! h_{11} h_{23} \!+\! h_{21} h_{13}) \Lambda_n (1 \!-\! \epsilon_1 Le)   \Bigr\} =0.
\end{equation}
Under the assumption $\epsilon_1 Le \leq 1$, from \eqref{PES3} it necessarily follows $\sigma_I=0$, i.e. $\bar{\sigma} \in \mathbb R$
\end{proof}

\begin{rem}\label{remark}
If we confine ourselves to the case of a single component fluid (i.e. for $\CC^2 \rightarrow 0$), we actually recover the model describing the evolutionary behaviour of a fluid-saturated anisotropic Brinkman bi-disperse porous medium, rotating about the vertical axis, see \cite{massa2}. 
In particular, \eqref{PES2} becomes
\begin{equation} \label{PES2-legg}
\!\!\! \Ra^2 \!\!=\!\! \dfrac{ (\Lambda_n \!+\! \sigma_R) (h_{11} h_{22} \!-\! h_{12} h_{21})}{a^2(h_{12} h_{23} \!-\! h_{13} h_{22}\!-\! h_{11} h_{23} \!+\! h_{21} h_{13})} + i \dfrac{ \sigma_I (h_{11} h_{22} \!-\! h_{12} h_{21})}{a^2(h_{12} h_{23} \!-\! h_{13} h_{22} \!-\! h_{11} h_{23} \!+\! h_{21} h_{13})}  \!\!
\end{equation}
therefore
\begin{equation} \label{PES3-legg}
\sigma_I \dfrac{ h_{11} h_{22} - h_{12} h_{21}}{a^2(h_{12} h_{23} - h_{13} h_{22} - h_{11} h_{23} + h_{21} h_{13})} =0,
\end{equation} 
From \eqref{PES3-legg} it follows $\sigma_I=0$, i.e. $\bar{\sigma} \in \mathbb R$ and {\it the strong form of the principle of exchange of stability holds}, under no additional hypotheses. Therefore, when there is no concentration gradient, convection can set in only through stationary motions.  
\end{rem}

\subsection{Steady convection threshold}\label{steady}
The marginal state for stationary convective instabilities is reached for $\bar\sigma=0$ $(\sigma_R=0,\sigma_I=0)$, so from $(\ref{PES1})$ we derive the critical Rayleigh number for the onset of stationary convection:
\begin{equation}\label{staz}
\Ra^2_S=\min_{(n,a^2) \in \N \times \R^{+}} \dfrac{\Lambda_n}{a^2} \dfrac{ h_{11} h_{22} \!-\! h_{12} h_{21}}{h_{12} h_{23} \!-\! h_{13} h_{22}\!-\! h_{11} h_{23} \!+\! h_{21} h_{13}} + \CC^2
\end{equation}
As already pointed out, if we consider a single component fluid (i.e. for $\CC^2 \rightarrow 0$), $(\ref{staz})$ coincides with the instability threshold found in \cite{massa2}.

\subsection{Oscillatory convection threshold}\label{oscill}
The marginal state for oscillatory convection is characterized by $\bar\sigma=i \sigma_I$, $(\sigma_1 \in \R-\{ 0 \}, \sigma_R=0)$, so from $(\ref{PES1})$ and \eqref{PES3} it follows
\begin{equation}\label{det2}
\Ra^2_O = \min_{(n,a^2) \in \N \times \R^{+}} \dfrac{\Lambda_n}{a^2} \dfrac{ h_{11} h_{22} \!-\! h_{12} h_{21}}{h_{12} h_{23} \!-\! h_{13} h_{22}\!-\! h_{11} h_{23} \!+\! h_{21} h_{13}} + \CC^2 \dfrac{\Lambda_n^2 + \epsilon_1 Le \sigma_I^2}{\Lambda_n^2 + (\epsilon_1 Le \sigma_I)^2}
\end{equation}
where the frequency of the oscillations $\sigma_I$ is given by
\begin{equation}
\sigma_I^2= \dfrac{a^2 (h_{12} h_{23} \!-\! h_{13} h_{22}\!-\! h_{11} h_{23} \!+\! h_{21} h_{13}) \Lambda_n \CC^2 (\epsilon_1 Le -1)-\Lambda_n^2 (h_{11} h_{22} \!-\! h_{12} h_{21})}{(\epsilon_1 Le)^2 (h_{11} h_{22} \!-\! h_{12} h_{21})}.
\end{equation}
Therefore, the linear instability threshold for the onset of oscillatory convection is 
\begin{equation}\label{osc}
\Ra^2_O = \min_{(n,a^2) \in \N \times \R^{+}} \dfrac{\Lambda_n}{a^2} \dfrac{ h_{11} h_{22} \!-\! h_{12} h_{21}}{h_{12} h_{23} \!-\! h_{13} h_{22}\!-\! h_{11} h_{23} \!+\! h_{21} h_{13}} \Bigl( 1 + \dfrac{1}{\epsilon_1 Le} \Bigr)+ \dfrac{\CC^2}{\epsilon_1 Le} 
\end{equation}
Let us underline that the relation between the steady and the oscillatory thresholds is given by
\begin{equation}\label{28}
\Ra_O^2=\Ra_S^2 \Bigl( 1 + \dfrac{1}{\epsilon_1 Le} \Bigr) -\CC^2,
\end{equation}
so for increasing $\CC^2$, i.e. for high salt concentrations, convection will arise via oscillatory motions.

\section{Results and Discussion}\label{num}
Due to the complicated algebraic form of the instability thresholds \eqref{staz} and \eqref{osc}, we perform numerical simulations via Matlab software in order to outline \emph{how} rotation, Brinkman model, anisotropy and concentration gradient affect the onset of convection, i.e. to outline the influence of the fundamental parameters $\T^2,Da_f,h,k,\CC^2$ on the steady and oscillatory instability thresholds \eqref{staz} and \eqref{osc}, respectively. In the following simulations, let us fix $\{\eta=0.2,\sigma=0.3,\gamma_1=0.9,\gamma_2=1.8, \epsilon_1 Le=55.924 \}$ (see \cite{BSsale, brian1, brian3, inertia}). \\ 
We numerically obtained that the minimum \eqref{staz} and \eqref{osc} with respect to $n$ is attained at $n=1$ and in Figure \ref{figRSRO} the neutral curves are shown, where we set 
\begin{equation}
\begin{aligned}
f_S^2(a^2) &= \dfrac{\Lambda_1}{a^2} \dfrac{ h_{11} h_{22} \!-\! h_{12} h_{21}}{h_{12} h_{23} \!-\! h_{13} h_{22}\!-\! h_{11} h_{23} \!+\! h_{21} h_{13}} + \CC^2, \\
f_O^2(a^2) &= \dfrac{\Lambda_1}{a^2} \dfrac{ h_{11} h_{22} \!-\! h_{12} h_{21}}{h_{12} h_{23} \!-\! h_{13} h_{22}\!-\! h_{11} h_{23} \!+\! h_{21} h_{13}} \Bigl( 1 + \dfrac{1}{\epsilon_1 Le} \Bigr)+ \dfrac{\CC^2}{\epsilon_1 Le} .
\end{aligned}
\end{equation}
\noindent
In Figures \ref{figT}(a) and \ref{figT}(b) the steady and oscillatory Rayleigh numbers $\Ra^2_S$ and $\Ra_O^2$ are depicted as functions of the Taylor number $\T^2$, we can conclude that the instability thresholds are increasing functions with respect to $\T^2$, so the rotation of the layer has a stabilizing effect on the onset of double-diffusive convection. In particular, the instability thresholds are represented for a low concentration Rayleigh number $\CC^2$ in Figure \ref{figT}(a) and for a high concentration Rayleigh number in Figure \ref{figT}(b): when the concentration gradient in the layer is low, convection sets in via stationary motions, but when the concentration gradient is high, oscillatory convection arises. \\ 
The asymptotic behaviour of the instability thresholds with respect to the Rayleigh number for the salt field is clearly depicted in Figure \ref{figC}: both $\Ra^2_S$ and $\Ra_O^2$ are linear and increasing function of $\CC^2$, so $(i)$ when a salt dissolved at the bottom of the layer is considered, the convection is delayed, $(ii)$ for increasing concentration Rayleigh numbers, double-diffusive convection occurs via oscillatory motions.   \\
%In Table \ref{tab_hk} the critical steady and oscillatory Rayleigh numbers $\Ra^2_S$ and $\Ra_O^2$ are shown for increasing quoted values of the micropermeability parameter $h$ and macropermeability parameter $k$.
In Tables \ref{tab_Da-hk}(a) and \ref{tab_Da-hk}(b) the combined effects that anisotropy and the Brinkman model have on the onset of double-diffusive convection are depicted. In particular, the critical steady and oscillatory Rayleigh numbers $\Ra^2_S$ and $\Ra_O^2$ are shown for increasing quoted values of the Darcy number $Da_f$ when the micropermeability parameter $h$ is lower - Table \ref{tab_Da-hk}(a) - and higher - \ref{tab_Da-hk}(b) - than the macropermeability parameter $k$. Both critical steady and oscillatory Rayleigh numbers increase as the Darcy number increases, i.e. $Da_f$ has a stabilizing effect on the onset of convection. Moreover, for very law $Da_f$, oscillatory convection occurs, while as $Da_f$ increases, there is a switch from oscillatory to steady convection. Let us finally observe that when $h<<k$, the instability thresholds are larger then the ones for the case $h>>k$, so when the micropermeability parameter is larger then the macropermeability parameter, the onset of convection is facilitated. This behaviour is depicted also in Figures \ref{fig_Da-hk}(a) and \ref{fig_Da-hk}(b).

\begin{figure}[h!]
\centering
\includegraphics[scale=0.45]{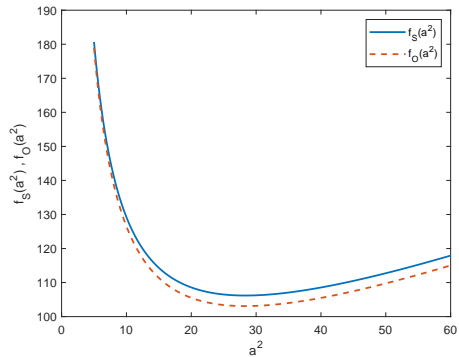}
\caption{Neutral curves for $h=0.1,k=10,\T^2=10,Da=0.001,\CC^2=5$.}
\label{figRSRO}
\end{figure}

\begin{figure}[h!]
\centering
\subfigure[$\CC^2=1.5$]{
\includegraphics[scale=0.45]{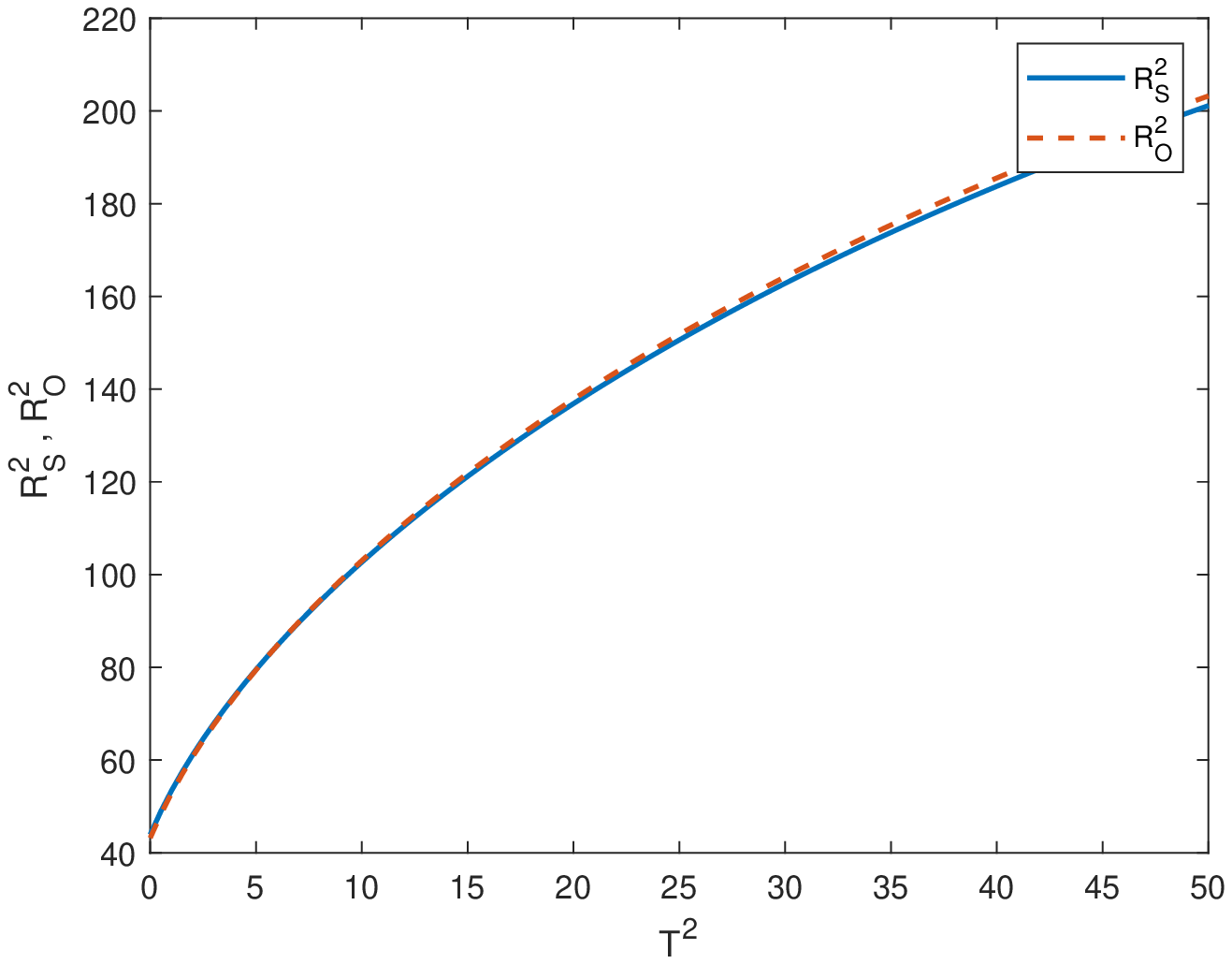}}
\subfigure[$\CC^2=5$]{
\includegraphics[scale=0.45]{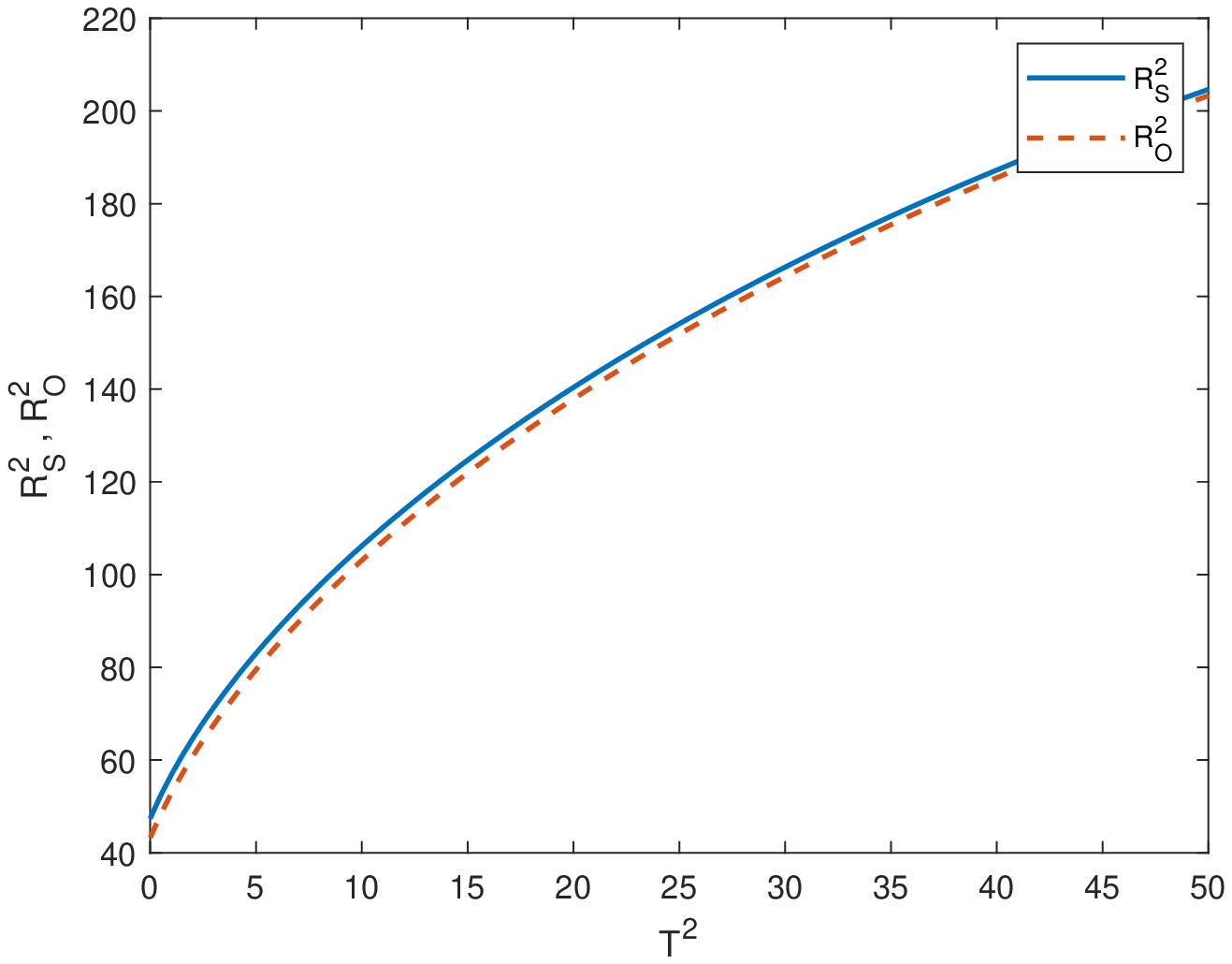}}
\caption{(a): Asymptotic behaviour of $\Ra^2_S$ and $\Ra_O^2$ with respect to $\T^2$ for $h=0.1,k=10,\CC=1.5,Da=0.001$. (b): asymptotic behaviour of $\Ra^2_S$ and $\Ra_O^2$ with respect to $\T^2$ for $h=0.1,k=10,\CC=5,Da=0.001$.}
\label{figT}
\end{figure}

\begin{figure}[h!]
\centering
\includegraphics[scale=0.45]{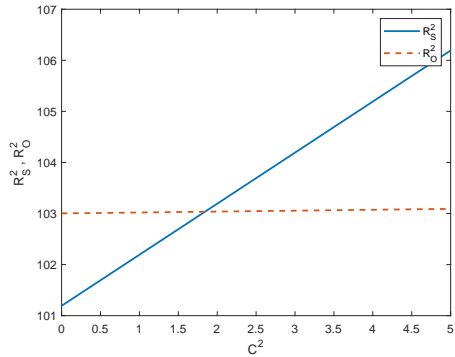}
\caption{Asymptotic behaviour of $\Ra^2_S$ and $\Ra_O^2$ with respect to $\CC^2$ for $h=0.1,k=10,\T^2=10,Da=0.001$.}
\label{figC}
\end{figure}

%\begin{table}[h!]
%\centering
%\subtable[$k=1$]{
%\begin{tabular}{|c|c|c|}
%\hline
%$R^2_S$ & $R^2_O$ & $h$ \\
%\hline
%91.8621 & 88.5048 & 0.1 \\
%56.5218 & 52.5325 & 1 \\
%55.5987 & 51.5929 & 1.5 \\
%58.4219 & 54.4666 & 5 \\
%61.2192 & 57.3139 & 10 \\
%\hline
%\end{tabular}
%}
%\qquad
%\subtable[$h=1$]{
%\begin{tabular}{|c|c|c|}
%\hline
%$R^2_S$ & $R^2_O$ & $k$ \\
%\hline
%64.0531 & 60.1984 & 0.1 \\
%56.5218 & 52.5325 & 1 \\
%a & a & 1.5 \\
%a & a & 5 \\
%71.5407 & 67.8199 & 10 \\
%\hline
%\end{tabular}
%}
%\caption{(a): Critical steady and oscillatory Rayleigh numbers for increasing micropermeability parameter $h$. (b): Critical steady and oscillatory Rayleigh numbers for increasing macropermeability parameter $k$. The other parameters are $\T^2=10, \CC^2=5, Da_f=0.001$.}
%\label{tab_hk}
%\end{table}
% 

\begin{table}[h!]
\centering
\subtable[$h<<k$]{
\begin{tabular}{|c|c|c|}
\hline
$R^2_S$ & $R^2_O$ & $Da_f$ \\
\hline
106.1926 & 103.0914 & 0.001 \\
408.1480 & 410.4462 & 1 \\
1615.8 & 1639.7 & 5 \\
\hline
\end{tabular}
}
\qquad
\subtable[$h>>k$]{
\begin{tabular}{|c|c|c|}
\hline
$R^2_S$ & $R^2_O$ & $Da_f$ \\
\hline
54.0168 & 49.9827 & 0.001 \\
369.6938 & 371.3044 & 1 \\
1560 & 1582.9 & 5 \\
\hline
\end{tabular}
}
\caption{(a): Critical steady and oscillatory Rayleigh numbers for increasing Darcy number $Da_f$ for $h=0.1,k=10$. (b): Critical steady and oscillatory Rayleigh numbers for increasing Darcy number $Da_f$ for $h=10,k=0.1$. The other parameters are $\T^2=10, \CC^2=5$.}
\label{tab_Da-hk}
\end{table}
 
\begin{figure}[h!]
\centering
\subfigure[$Da_f=0.001$]{
\includegraphics[scale=0.45]{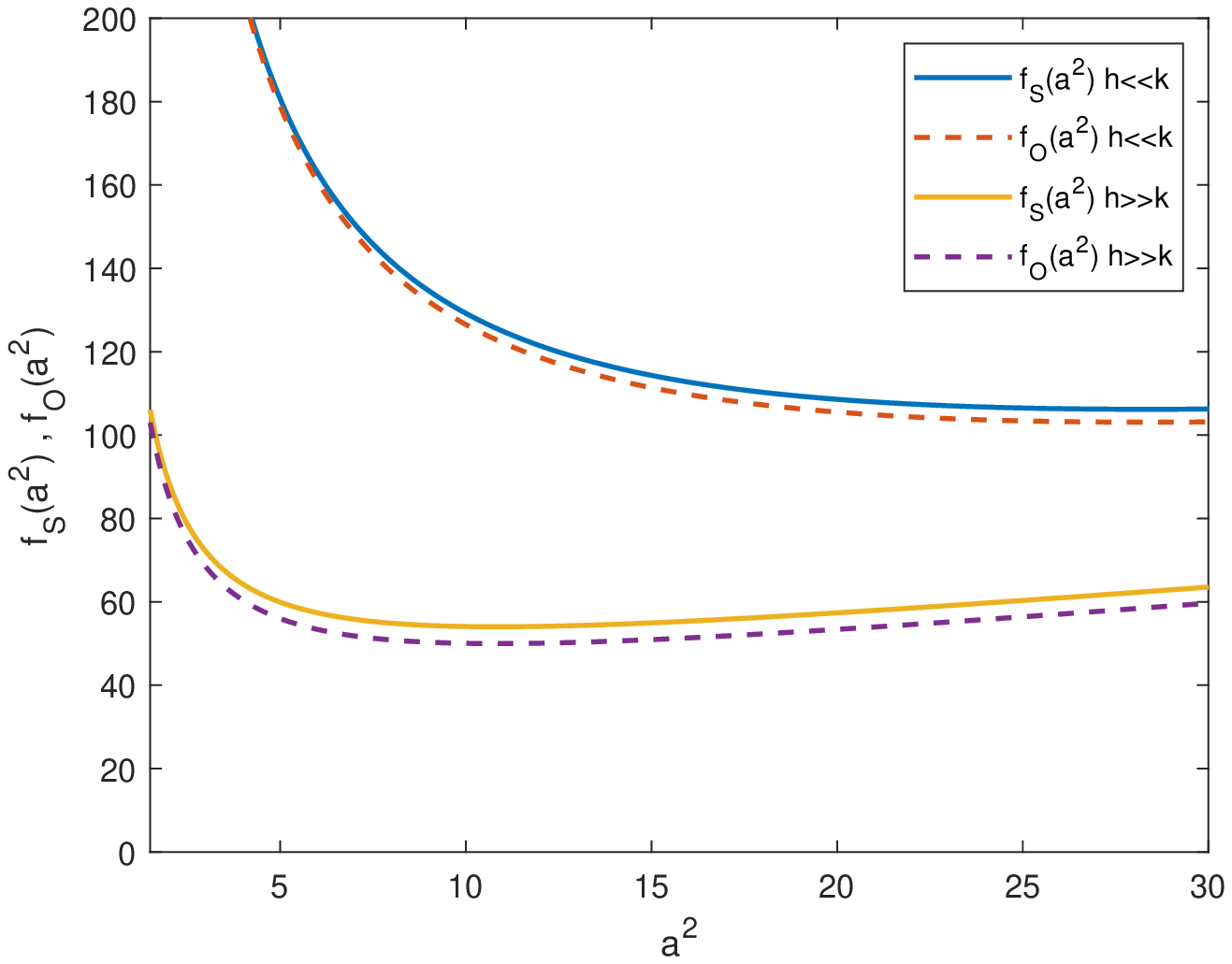}}
\subfigure[$Da_f=1$]{
\includegraphics[scale=0.45]{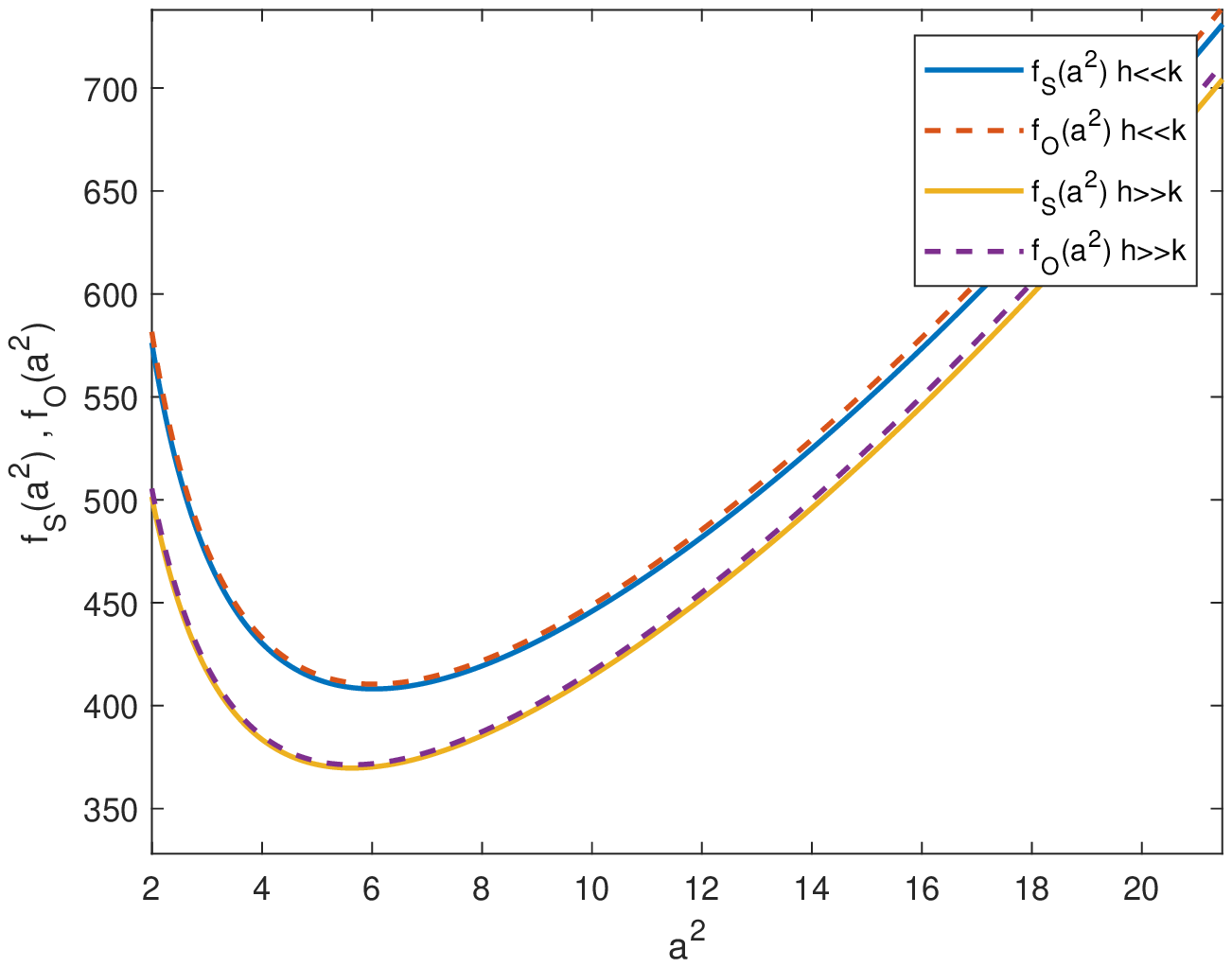}}
\caption{(a): Neutral curves at $Da_f=0.001$. (b): Neutral curves at $Da_f=1$. \\ The other parameters are $\T^2=10,\CC^2=5$. The case $h<<k$ is $h=0.1,k=10$, while the case $h>>k$ is $h=10,k=0.1$  }
\label{fig_Da-hk}
\end{figure}

\section{Conclusions}\label{concl}
In this paper, the onset of convection in a rotating horizontal layer of anisotropic bi-disperse porous material simultaneously heated and salted from below was analysed. We determined the instability thresholds for the onset of double-diffusive convection via steady and oscillatory motions. Moreover, we proved the validity of the principle of exchange of stabilities under the assumption $\epsilon_1 Le \leq 1$, so in this case only stationary convection can occur. Numerical simulations were performed in order to analyse the behaviour of the instability thresholds with respect to the fundamental parameters, in particular we found that rotation and concentration gradient act to delay the onset of convection.

\textbf{Acknowledgements.} This paper has been performed under the auspices of the GNFM of INdAM. \\ R. De Luca and G. Massa would like to thank Progetto Giovani GNFM 2020: "Problemi di convezione in nanofluidi e in mezzi porosi bidispersivi".

\bibliographystyle{unsrt}

\bibliography{BDPM_Acta-Sale}

\end{document}